\documentclass{article}%
\usepackage{amsfonts}
\usepackage{amsmath}
\usepackage{amssymb}
\usepackage{graphicx}%
\providecommand{\U}[1]{\protect\rule{.1in}{.1in}}
\providecommand{\U}[1]{\protect\rule{.1in}{.1in}}
\newtheorem{theorem}{Theorem}

\newtheorem{corollary}[theorem]{Corollary}

\newtheorem{proposition}[theorem]{Proposition}

\newenvironment{proof}[1][Proof]{\noindent\textbf{#1.} }{\ \rule{0.5em}{0.5em}}
\begin{document}

\title{Flat coordinates of flat St\"{a}ckel systems}
\author{Krzysztof Marciniak\\Department of Science and Technology \\Campus Norrk\"{o}ping, Link\"{o}ping University\\601-74 Norrk\"{o}ping, Sweden\\krzma@itn.liu.se
\and Maciej B\l aszak\\Faculty of Physics, Division of Mathematical Physics, A. Mickiewicz University\\Umultowska 85, 61-614 Pozna\'{n}, Poland\\blaszakm@amu.edu.pl}
\maketitle

\begin{abstract}
In this article we explicitely construct transformation bewteen separable and
flat coordinates for flat St\"{a}ckel systems and exploit the structre of
these systems in flat coordinates. In the elliptic case these coordinates
become well known generalized elliptical coordinates of Jacobi.

\end{abstract}


\section{Introduction}

The search for flat coordinates for systems that we a priori know are flat is
not easy. This article is devoted to search for flat coordinates for the so
called St\"{a}ckel systems \cite{Stackel}. St\"{a}ckel systems are roughly
speaking (for more precise definition, see below) Hamiltonian systems
separable in the sense of Hamilton-Jacobi theory by a pointwise transformation
to orthogonal coordinates. As such, they are of great importance in theory of
classical integrable systems.

In this paper we construct separable flat systems of St\"{a}ckel type directly
from scratch i.e. from an appropriate separation curve (or an appropriate set
of separation relations \cite{Sklyanin}) and then find flat coordinates for
(almost) all flat St\"{a}ckel systems of Benenti type. We also establish the
signature of metric tensors of these systems. Further, we present the explicit
form of many important geometric objects connected to these flat St\"{a}ckel
systems (namely metric tensors, Killing tensors and separable potentials) in
these new coordinates. Thus, we end up with separable flat Hamiltonians
written in flat coordinates of respective pseudo-Euclidian metrices.

Our construction encompasses two known cases: Jacobi elliptic coordinates
(introduced in \cite{Jacobi1} and fully described in \cite{Jacobi2}) and
Jacobi parabolic coordinates and also one of the less known cases considered
recently by Blaszak and Sergyeyev in \ \cite{blaser} (but with no degeneration
of coordinate systems).

The paper is organized as follows. In Section 2 we remind basic facts about
St\"{a}ckel systems and in particular about St\"{a}ckel systems of Benenti
type. In our approach we constructing St\"{a}ckel systems directly in their
separation coordinates\ using an appropriate separation relations (separation
curve). In Section 3 we present the construction of flat coordinates in the
case of real roots in the polynomial that defines a given Benenti system.
Section 4 is devoted to expressing various tensor objects in our coordinates
and given known and new formulas for a variety of separable potentials.
Finally, in Section 5 we consider the case of complex conjugate (but still
nondegenerate) roots. The case of degenerated roots is non studied in this paper.

\section{St\"{a}ckel systems}

Consider a set of Darboux coordinates (often called canonical coordinates)
$(\lambda,\mu)=(\lambda_{1}\ldots,,\lambda_{n},\mu_{1},\ldots,\mu_{1}) $ on a
$2n$-dimensional Poisson manifold $M$ equipped with a Poisson operator $\Pi$
(so that $\Pi=%
{\textstyle\sum\nolimits_{i=1}^{n}}
\frac{\partial}{\partial\lambda_{i}}\wedge\frac{\partial}{\partial\mu_{i}}$).
A classical St\"{a}ckel system on $M$ is a system of $n$ Hamiltonians (i.e.
smooth real-valued functions) $H_{i}$ defined on a dense open subset of $M$
originating from a set of $n$ separation relations \cite{Sklyanin} of the
form:
\begin{equation}
\sigma(\lambda_{i})+%
{\displaystyle\sum\limits_{j=1}^{n}}
H_{j}\lambda_{i}^{\gamma_{j}}=\frac{1}{2}f(\lambda_{i})\mu_{i}^{2}\text{,
\ \ \ }i=1,\ldots,n,\label{Stack}%
\end{equation}
where $f$ and $\sigma$ are arbitrary functions of one argument and where all
$\gamma_{i}\in\mathbf{Z,}$ $i=1,\ldots,n,$ and are such that no two
$\gamma_{i}$ coincide. Thus, a particular St\"{a}ckel system is defined by the
choice of integers $\gamma_{1},\ldots,\gamma_{n}$ and by the choice of
functions $f$ and $\sigma$. Customary one can also treat this system of
relations as $n$ points on ($n$ copies of) the following \emph{separation
curve}%
\begin{equation}
\sigma(\lambda)+%
{\displaystyle\sum\limits_{j=1}^{n}}
H_{j}\lambda^{\gamma_{j}}=\frac{1}{2}f(\lambda)\mu^{2},\label{Stackc}%
\end{equation}
in $\lambda\mu$ plane which helps us to avoid writing too many indices. The
relations (\ref{Stack}) (or $n$ copies of (\ref{Stackc})) constitute a system
of $n$ equations linear in the unknowns $H_{i}$. Solving these relations with
respect to $H_{i}$ we obtain $n$ functions $H_{i}=H_{i}(\lambda,\mu)$ on $M$
commuting (since the right-hand sides of formulas (\ref{Stack}) commute) with
respect to the Poisson operator $\Pi$:%
\[
\left\{  H_{i},H_{j}\right\}  _{\Pi}\equiv\Pi(dH_{i},dH_{j})=0\text{ for all
}i,j=1,\ldots,n
\]
These functions have the form
\begin{equation}
H_{i}=\frac{1}{2}\mu^{T}K_{i}G\mu+V_{i}(\lambda)\text{ \ }i=1,\ldots
,n\text{,}\label{Stackham}%
\end{equation}
where we denote $\lambda=(\lambda_{1},\ldots,\lambda_{n})^{T}$ and $\mu
=(\mu_{1},\ldots,\mu_{n})^{T}$. The functions $H_{i}$ can be interpreted as
$n$ quadratic in momenta $\mu$ Hamiltonians on the phase space $M=T^{\ast
}\mathcal{Q}$ cotangent to a Riemannian manifold $\mathcal{Q}$ (so that
$\lambda_{1,}\ldots,\lambda_{n}$ are coordinates on $\mathcal{Q}$) equipped
with the contravariant metric tensor $G$ depending on the function $f$ \ and
the choice of the constants $\gamma_{i}$. They are commonly known as
\emph{St\"{a}ckel Hamiltonians} on $M$. Note also that by the very
construction of $H_{i}$ the variables $\left(  \lambda,\mu\right)  $ are
separation variables for all the Hamiltonians in (\ref{Stackham}) in the sense
that the Hamilton-Jacobi equations associated with the Hamiltonians $H_{i}$
admit a common additively separable solution $W=%
{\textstyle\sum_{i=1}^{n}}
W_{i}(\lambda_{i},a)$. Further, the objects $K_{i}$ in (\ref{Stackham}) can be
interpreted as $(1,1)$-type Killing tensors on $\mathcal{Q}$ for the metric
$G$. The metric tensor $G$ and all the Killing tensors $K_{i}$ \ in
(\ref{Stackham}) are diagonal in $\lambda$-variables and it is easy to see
that the Killing tensors $K_{i}$ do not depend neither on a particular choice
of $f$ nor $\sigma$: changing $\sigma$ we change the potentials $V_{i}%
(\lambda)$ while changes of $f$ influence the metric $G$. We define also a
$(2,0)$-type tensors $A_{i}$ (contravariant Killing tensors) by
\begin{equation}
A_{i}=K_{i}G\text{, \ \ }i=1\ldots n\label{A}%
\end{equation}
so that since $K_{1}=I$ we have $A_{1}=G$.

A particular subclass of St\"{a}ckel systems is given by choosing the
separation curve (\ref{Stackc}) in the form%
\begin{equation}%
{\displaystyle\sum\limits_{j=1}^{n}}
H_{j}\lambda^{n-j}=B_{m}(\lambda)\left(  \frac{1}{2}\mu^{2}+\lambda
^{k}\right)  \text{, \ \ \ }m\in\mathbf{N}\text{, }k\in\mathbf{Z}\label{BenSC}%
\end{equation}
(so that $f(\lambda)=B_{m}(\lambda)$ while $\sigma(\lambda)=-\lambda^{k}%
B_{m}(\lambda)$ with an arbitrary fixed integer $k$) where%
\[
B_{m}(\lambda)=\sum_{j=0}^{m}\lambda^{m-j}\rho_{j}^{(m)}(\beta)\equiv
\prod\limits_{j=1}^{m}\left(  \lambda-\beta_{j}\right)
\]
is a real polynomial of order $m$ in $\lambda$ with possibly complex roots
$\beta_{j}$ (so that $\beta_{j}$ are either real or exist in complex conjugate
pairs) that are all assumed to be different (we assume trhoughout the article
that there is no degeneracy in the roots of the polynomial $B_{m} $). The real
coefficients $\rho_{j}^{(m)}(\beta)$ are thus Vi\`{e}te polynomials (signed
symmetric polynomials) of the possibly complex constants $\beta_{1}%
,\ldots,\beta_{m}$:
\begin{equation}
\rho_{j}^{(m)}(\beta)=(-1)^{j}%
{\displaystyle\sum\limits_{1\leq s_{1}<s_{2}<\ldots<s_{j}\leq m}}
\beta_{s_{1}}\ldots\beta_{s_{j}}\text{, \ \ }j=1,\ldots,m\label{defrho}%
\end{equation}
and in case of no ambiguity (when $m$ is obvious) we will simply denote them
as $\rho_{j}$. The Hamiltonians $H_{i}$ generated by the separation curve
(\ref{BenSC}) constitute a completely integrable system that is called a
St\"{a}ckel system of \emph{Benenti type} (or simply a \emph{Benenti system})
due to S. Benenti's contribution to the study of these objects \cite{Ben1}%
,\cite{Ben2}. The Hamiltonians $H_{i}$ have the form (\ref{Stackham}) with the
metric tensor $G$ and the Killing tensors $K_{i}$ given explicitely through%
\begin{equation}
G=\operatorname*{diag}\left(  \frac{f(\lambda_{1})}{\Delta_{1}},\ldots
,\frac{f(\lambda_{n})}{\Delta_{n}}\right)  =\operatorname*{diag}\left(
\frac{B_{m}(\lambda_{1})}{\Delta_{1}},\ldots,\frac{B_{m}(\lambda_{n})}%
{\Delta_{n}}\right)  ,\text{ \ \ }\Delta_{i}=%
{\textstyle\prod\limits_{j\neq i}}
(\lambda_{i}-\lambda_{j})\label{GBen}%
\end{equation}%
\begin{equation}
K_{i}=-\operatorname*{diag}\left(  \frac{\partial q_{i}}{\partial\lambda_{1}%
},\cdots,\frac{\partial q_{i}}{\partial\lambda_{n}}\right)  \text{
\ \ \ }i=1,\ldots,n\label{Ki}%
\end{equation}
Here and below $q_{i}=q_{i}(\lambda)$ are Vi\`{e}te polynomials in the
variables $\lambda_{1},\ldots,\lambda_{n}$:%
\begin{equation}
q_{i}(\lambda)=(-1)^{i}%
{\displaystyle\sum\limits_{1\leq s_{1}<s_{2}<\ldots<s_{i}\leq n}}
\lambda_{s_{1}}\ldots\lambda_{s_{i}}\text{, \ \ }i=1,\ldots,n\label{defq}%
\end{equation}
(cf (\ref{defrho})) that can also be considered as new coordinates on\ the
Riemannian manifold $\mathcal{Q}$ (we will then refer to them as Vi\`{e}te coordinates).

\begin{proposition}
The metric (\ref{GBen}) is flat only for $m\leq n$ and is of constant
curvature for $m=n+1$. For higher $m$ it has a non-constant curvature.
\end{proposition}

One proves this by direct calculation of scalar curvature of (\ref{GBen}). The
above proposition means that it is meaningful to seek for flat coordinates for
Beneti systems only in case when $m=0,\ldots,n$.

Let us now turn our attention to the separable potentials $V_{i}(\lambda)$ in
(\ref{Stackham}) in Benenti case. If we remove the $\lambda^{k}$ term from the
right hand side of (\ref{BenSC}) we receive a geodesic Benenti system (with
all potentials in (\ref{Stackham}) equal to zero). In the non-geodesic case
(that is the case generated by the full separation curve (\ref{BenSC})) the
potentials $V_{i}(\lambda)$ depend on the constants $m$ and $k$ (as well as on
the dimension $n$) so we will denote them by $V_{i}^{(m,k)}(\lambda)$ or
simply by $V_{i}^{(m,k)}$. Notice again that these potentials are generated by
the term $\sigma(\lambda)=-\lambda^{k}B_{m}(\lambda)$ in the separation curve
(\ref{BenSC}). Further by $V^{(m,k)}$ we will denote the column vector with
components $V_{i}^{(m,k)}$ so that%
\[
V^{(m,k)}=\left(  V_{1}^{(m,k)},\ldots,V_{n}^{(m,k)}\right)  ^{T}%
\]
By solving (\ref{BenSC}) with respect to $H_{i}$ one obtains that%
\begin{equation}
V^{(m,k)}=\sum_{j=0}^{m}\rho_{j}^{(m)}(\beta)U^{(m-j+k)}\label{V}%
\end{equation}
where the column vector $U^{(k)}$ represents the so called basic separable
potentials related to $\sigma(\lambda)=-\lambda^{k}$ which can be constructed
recursively \cite{R1,R2} by%
\begin{equation}
U^{(k)}=R^{k}U^{(0)}\label{U}%
\end{equation}
with the recursion matrix $R$ of the form%

\begin{equation}
R=\left(
\begin{array}
[c]{cccc}%
-q_{1} & 1 &  & \\
-q_{2} &  & \ddots & \\
\vdots &  &  & 1\\
-q_{n} & 0 & \cdots & 0
\end{array}
\right) \label{R}%
\end{equation}
and with $U^{(0)}=(0,0,\ldots,0,1)^{T}$. Note that the formulas (\ref{V}%
)-(\ref{R}) are non tensor in that they are the same in an arbitrary
coordinate system, not only in the separation variables $\lambda_{i}$. Note
also that for $m=0$ we have $V_{r}^{(0,k)}=U_{r}^{(k)}$ so that for $m=0$ both
families of potentials coincide. The potentials $V$ are naturally linear
combinations of the basic separable potentials $U$ determined by our specific
choice of the function $\sigma(\lambda)$ in (\ref{BenSC}). This choice is
motivated by the fact that the potentials $V$ in flat coordinates generalize
the well known potentials as it will be demonstrated below. The "lowest" basic
separable potentials have the following form: $U^{(1)}=RU^{(0)}=(0,0,\ldots
0,1,0)^{T}$ up to $U^{(n-1)}=R^{n-1}U^{(0)}=(1,0,\ldots,0)^{T}$ are trivial
(constant), $U^{(n)}=R^{n}U^{(0)}=(-q_{1},\ldots,-q_{n})$ is the first
nontrivial positive potential while $U^{(-1)}=R^{-1}U^{(0)}$ $=(1/q_{n}%
,q_{1}/q_{n},\ldots,q_{n-1}/q_{n})^{T}$. The "negative" potentials (i.e.
potentials obtained for negative $k$) are rational functions of $q$ that
quickly become complicated with decreasing $k$.

More information on Benenti systems can be found in
\cite{bensol,bensol2,bensol3}.

\section{Flat coordinates for St\"{a}ckel systems - real case}

As we mentioned before, if we restrict ourselves to the case $0\leq m\leq n$
then the metirc $G$ in (\ref{GBen}) is flat so there is a legitimate question
of finding flat coordinates for this metric. In this section we construct flat
coordinates of $G$ in case where all the roots $\beta_{j}$ of $B_{m}(\lambda)$
are real. So, our aim is to find flat coordinates for the metric tensor $G$
for an arbitrary $m$ between $0$ an $n$ and for arbitrary real constants
$\beta_{1},\ldots,\beta_{m}$.

Consider thus the following generating function%
\begin{equation}%
{\displaystyle\sum\limits_{j=0}^{n-m}}
z^{n-m-j}a_{j}-\frac{1}{4}\varepsilon\sum_{j=1}^{m}\frac{x_{j}^{2}}%
{z-\beta_{j}}\equiv\frac{%
{\displaystyle\prod\limits_{j=1}^{n}}
(z-\lambda_{j})}{%
{\displaystyle\prod\limits_{j=1}^{m}}
(z-\beta_{j})}\label{GF}%
\end{equation}
(where $\varepsilon=+1$ or $\varepsilon=-1$ and where the identity is taken
with respect to the variable $z$). This function defines (locally) an
invertible map between variables $\left(  \lambda_{1},\ldots,\lambda
_{n}\right)  $ and \ new variables $\left(  x_{1},\ldots,x_{m},a_{1}%
,\ldots,a_{n-m}\right)  $ on our manifold $\mathcal{Q}$ whereas the choice of
the sign of $\varepsilon$ is governed by the actual sign of the variables in a
given region of our Riemannian manifold $\mathcal{Q}$. An easy way to see this
is to multiply both sides of (\ref{GF}) by $B_{m}(z)\equiv%
{\textstyle\prod\limits_{j=1}^{m}}
(z-\beta_{j})$ and compare the coefficients of polynomials on both sides of
the equation. We can see that $a_{0}=1$ in the above formula, so\ in case
$m=n$ the generating function (\ref{GF}) attains the form%
\[
1-\frac{1}{4}\varepsilon\sum_{j=1}^{n}\frac{x_{j}^{2}}{z-\beta_{j}}\equiv
\frac{%
{\displaystyle\prod\limits_{j=1}^{n}}
(z-\lambda_{j})}{%
{\displaystyle\prod\limits_{j=1}^{n}}
(z-\beta_{j})}%
\]
which in the regions of the manifold $\mathcal{Q}$ when $\varepsilon<0$ is
nothing else as the well known transformation (see \cite{Jacobi1} and
\cite{Jacobi2}) between the coordinates $\left(  x_{1},\ldots,x_{n}\right)  $
and the Jacobi elliptic coordinates $\left(  \lambda_{1},\ldots,\lambda
_{n}\right)  $. In the case $m=n-1$ the function (\ref{GF}) becomes
\[
z+a_{1}-\frac{1}{4}\varepsilon\sum_{j=1}^{n-1}\frac{x_{j}^{2}}{z-\beta_{j}%
}\equiv\frac{%
{\displaystyle\prod\limits_{j=1}^{n}}
(z-\lambda_{j})}{%
{\displaystyle\prod\limits_{j=1}^{m}}
(z-\beta_{j})}%
\]
which is commonly known as the generating function for transformation between
the coordinates $\left(  x_{1},\ldots,x_{n-1},a_{1}\right)  $ and the Jacobi
parabolic coordinates $\left(  \lambda_{1},\ldots,\lambda_{n}\right)  $. In
the case $m=0$ we consider instead of (\ref{GF}) the generating function of
the form%
\begin{equation}%
{\displaystyle\sum\limits_{j=0}^{n}}
z^{n-j}a_{j}\equiv%
{\displaystyle\prod\limits_{j=1}^{n}}
(z-\lambda_{j})\label{GF0}%
\end{equation}
so that $a_{i}(\lambda)=q_{i}(\lambda)$ i.e. the variables $\left(
a_{1},\ldots,a_{n}\right)  $ coincide then with the Vi\`{e}te coordinates
(\ref{defq}) while the variables $x_{i}$ are not present at all. One can say
that this function is a variant of (\ref{GF}) with both $\varepsilon$ and all
$\beta_{i}$ non-present.

Let us now investigate the map between coordinates $\left(  \lambda_{1}%
,\ldots,\lambda_{n}\right)  $ and $\left(  x_{1},\ldots,x_{m},a_{1}%
,\ldots,a_{n-m}\right)  $.

\begin{theorem}
The map from coordinates $\left(  \lambda_{1},\ldots,\lambda_{n}\right)  $ to
$\left(  x_{1},\ldots,x_{m},a_{1},\ldots,a_{n-m}\right)  $ is given by%
\begin{align}
x_{j}^{2}  & =-4\varepsilon\frac{%
{\displaystyle\prod\limits_{k=1}^{n}}
(\beta_{j}-\lambda_{k})}{%
{\displaystyle\prod\limits_{\substack{k=1 \\k\neq j}}^{m}}
(\beta_{j}-\beta_{k})}\text{, \ \ }j=1,\ldots m\label{xj}\\
\left(
\begin{array}
[c]{c}%
a_{1}\\
\vdots\\
a_{n-m}%
\end{array}
\right)   & =M\left(
\begin{array}
[c]{c}%
q_{1}(\lambda)-\rho_{1}(\beta)\\
\vdots\\
q_{n-m}(\lambda)-\rho_{n-m}(\beta)
\end{array}
\right) \label{aj}%
\end{align}
where $M$ is a square matrix with entries given by%
\begin{equation}
M_{ij}=\left\{
\begin{array}
[c]{cc}%
U_{1}^{(m,m-1+i-j)} & \text{for }j\leq i\\
0 & \text{for }j>i
\end{array}
\right.  \text{ \ with }i,j=1,\ldots,n-m\label{M}%
\end{equation}
where $U_{1}^{(m,m-1+i-j)}$ are basic separable potential given by (\ref{U})
with the dimension $n$ replaced by $m$.
\end{theorem}

\begin{proof}
To show (\ref{xj}) let us first multiply both sides of (\ref{GF}) by
$B_{m}(z)=%
{\textstyle\prod\limits_{k=1}^{m}}
(z-\beta_{k})$. We receive%
\begin{equation}
B_{m}(z)%
{\displaystyle\sum\limits_{k=0}^{n-m}}
z^{n-m-k}a_{k}-\frac{1}{4}\varepsilon B_{m}(z)\sum_{k=1}^{m}\frac{x_{k}^{2}%
}{z-\beta_{k}}\equiv%
{\displaystyle\prod\limits_{k=1}^{n}}
(z-\lambda_{k})\label{porownaj}%
\end{equation}
Let us now insert $z=\beta_{j}$ in (\ref{porownaj}). Since $B_{m}(\beta
_{j})=0$ we obtain
\[
-\frac{1}{4}\varepsilon x_{j}^{2}%
{\displaystyle\prod\limits_{\substack{k=1 \\k\neq j}}^{m}}
(\beta_{j}-\beta_{k})=%
{\displaystyle\prod\limits_{k=1}^{n}}
(\beta_{j}-\lambda_{k})
\]
from which (and since $1/\varepsilon=\varepsilon$) we obtain (\ref{xj}). The
formula (\ref{aj}) can be obtained by a careful comparison of coefficients of
polynomials in (\ref{porownaj}).
\end{proof}

By direct comparison of the coefficients in (\ref{porownaj}) one can also show
that%
\begin{equation}
q_{i}=\sum_{j=0}^{n-m}\rho_{i-j}a_{j}+\frac{1}{4}\varepsilon\sum_{j=1}%
^{m}\frac{\partial\rho_{i-(n-m)}}{\partial\beta_{j}}x_{j}^{2}\text{,
\ }i=1,\ldots,n\label{xrV}%
\end{equation}
which gives us the map from the variables $\left(  x_{1},\ldots,x_{m}%
,a_{1},\ldots a_{n-m}\right)  $ to the Vi\`{e}te variables (\ref{defq}). In
the above formula we use the notation $\rho_{i}=0$ for $i<0$ or for $i>m$ and
$\rho_{0}=1$.

Let us now turn to the problem of finding flat coordinates for the metric $G$
generated by (\ref{BenSC}). Consider the polynomial map (compare with formula
(12) in \cite{blaser})%
\begin{equation}
a_{i}=r_{i}+\frac{1}{4}\sum\limits_{j=1}^{i-1}r_{j}r_{i-j},\text{
\ \ }i=1,\ldots,n-m\label{ar}%
\end{equation}
from the variables $\left(  r_{1},\ldots,r_{n-m}\right)  $ to $\left(
a_{1},\ldots,a_{n-m}\right)  $. This map (\ref{ar}) is injective due to its
triangular structure. By theorem of Bialynicki-Birula and Rosenlicht
\cite{BBR} it is then also surjective and therefore bijective and since its
Jacobian, as it is easy to see from (\ref{ar}), is equal to $1$, by a variant
of Jacobian Conjecture (see for example \cite{Kurdyka} or \cite{Bass}) the map
inverse to (\ref{ar}) is also a polynomial map. Thus, we conclude that

\begin{proposition}
The transformation (\ref{ar}) is bijective and its inverse is also a
polynomial map.
\end{proposition}

Combining the maps (\ref{xj})-(\ref{aj}) and (\ref{ar}) we obtain the map
between the variables $(\lambda_{1},\ldots,\lambda_{n})$ and $\left(
x_{1},\ldots,x_{m},r_{1},\ldots r_{n-m}\right)  $. We are now in position to
formulate the main theorem of this paper.

\begin{theorem}
\label{mainthm}The metric $G$ defined by (\ref{GBen}) attains in coordinates
$\left(  x_{1},\ldots,x_{m},r_{1},\ldots r_{n-m}\right)  $ the form%
\[
G=\left(
\begin{array}
[c]{cc}%
\varepsilon I_{m\times m} & 0_{m\times(n-m)}\\
0_{(n-m)\times m} & J_{(n-m)\times(n-m)}%
\end{array}
\right)
\]
where $I_{k\times k}$ denotes the $k\times k$ identity matrix and $J_{k\times
k}$ denotes the $k\times k$ matrix given by $\left(  J_{k\times k}\right)
_{ij}=\delta_{i,k-j+1}$ i.e. with entries equal to zero everywhere except on
the antidiagonal where all the entries are equal to $1$. Naturally,
$0_{k\times r}$ denotes the $k\times r$ zero matrix.
\end{theorem}

\begin{proof}
One can show this theorem by directly calculating the Jacobian of the map
(\ref{xj})-(\ref{aj})-(\ref{ar}) but this yields a very tedious calculation.
Alternatively, by solving (\ref{BenSC}) with respect to $H_{i}$ we see that%
\[
G=\sum_{k=0}^{m}\rho_{k}^{(m)}G_{k}%
\]
(compare with (\ref{V})). The form of tensors $G_{k}$ in Vi\`{e}te coordinates
has been found in \cite{blaser} so one can easily first transform the tensor
$G$ to the Vi\`{e}te coordinates, and then to use the inverse of the Jacobian
of the map (\ref{xrV}) to transform it to the variables $(x,r)$.
\end{proof}

Thus, the variables $\left(  x_{1},\ldots,x_{m},r_{1},\ldots r_{n-m}\right)  $
are flat but non-orthogonal coordinates for the metric $G$ (while the
separation coordinates $\left(  \lambda_{1},\ldots,\lambda_{n}\right)  $ are
orthogonal but not flat, see (\ref{GBen})). It is now elementary to find the
transformation form coordinates $\left(  x_{1},\ldots,x_{m},r_{1},\ldots
r_{n-m}\right)  $ to the pseudo-Euclidean coordinates for $G$. However, the
formulas for Killing tensors and potentials (see below) become much less
transparent in these coordinates and this is why we stop at flat coordinates
$\left(  x_{1},\ldots,x_{m},r_{1},\ldots r_{n-m}\right)  $.

\begin{corollary}
The signature $(n_{+},n_{-})$ (where $n_{+}$ and $n_{-}$ is the number of
positive respective negative eigenvalues of $G$) of the metric $G$ is (in the
real case) given by%
\begin{align*}
\left(  n_{+},n_{-}\right)   & =\left(  n-\left[  \frac{n-m}{2}\right]
,\left[  \frac{n-m}{2}\right]  \right)  \text{ \ in the region where
}\varepsilon=+1\\
\left(  n_{+},n_{-}\right)   & =\left(  n-m-\left[  \frac{n-m}{2}\right]
,\left[  \frac{n-m}{2}\right]  +m\right)  \text{ \ in the region where
}\varepsilon=-1
\end{align*}
where $\left[  \alpha\right]  $ denotes the integer part of the number
$\alpha$.
\end{corollary}

This corollary means that the metric $G$ is Euclidean (in the appropriate
regions, where $\varepsilon=+1$) only in the elliptic case and in the
parabolic case (i.e. for $m=n$ and $m=n-1$), otherwise it is pseudo-Euclidean.
Note also that in case $m=0$ both expressions coincide.

We will now investigate the structure of the Killing tensors $A_{r}$ (defined
in (\ref{A})) and separable potentials $V$ (defined through formulas
(\ref{V}), (\ref{U}) and (\ref{R})) in the flat coordinates\ $\left(
x_{1},\ldots,x_{m},r_{1},\ldots r_{n-m}\right)  $ in the elliptic ($m=n$)
case, in the parabolic ($m=n-1$) case and in the case of $m=0$ (in the case of
arbitrary $m$ the formulas become very complicated and non-transparent). Let
us start with the elliptic case $m=n$. The form of the $(2,0)$-type tensors
$A_{r}$ in flat coordinates can be calculated by the usual transformation
rules for tensors. The result is presented below.

\begin{proposition}
For $m=n$ the $(2,0)$-tensors $A_{s}$ ($s=1,\ldots,n$) defined in (\ref{A})
attain in the flat coordinates $\left(  x_{1},\ldots,x_{n}\right)  $ the form%
\begin{align*}
A_{s}^{ij}  & =\frac{1}{4}\frac{\partial^{2}\rho_{s}}{\partial\beta
_{i}\partial\beta_{j}}x_{i}x_{j}\text{, \ \ }i\neq j\\
A_{s}^{ii}  & =-\varepsilon\frac{\partial\rho_{s}}{\partial\beta_{i}}-\frac
{1}{4}\sum\limits_{\substack{k=1 \\k\neq i}}^{n}\frac{\partial^{2}\rho_{s}%
}{\partial\beta_{i}\partial\beta_{k}}x_{k}^{2}%
\end{align*}
(no summation over repeated indices is performed here) where $\rho_{s}%
=\rho_{s}^{(n)}(\beta_{1},\ldots,\beta_{n})$ is given by (\ref{defrho}).
\end{proposition}

It is not possible to present the general formula for the potentials
$V_{r}^{(n,k)}$ in flat coordinates but we can at least present few potentials
with low $k$. We denote $x=(x_{1},\ldots,x_{n})^{T}$ and by $\left(
\cdot,\cdot\right)  $ we denote the usual scalar product in $\mathbf{R}^{n}$.
Further, denote
\begin{align*}
\Gamma_{s}  & =-\operatorname*{diag}\left(  \frac{\partial\rho_{s}^{(n)}%
}{\partial\beta_{1}},\cdots,\frac{\partial\rho_{s}^{(n)}}{\partial\beta_{n}%
}\right)  \text{ \ \ \ }s=1,\ldots,n\\
\Delta & =\operatorname*{diag}(\beta_{1},\ldots,\beta_{n})\\
W  & =1+\frac{1}{4}\varepsilon(x,\Delta^{-1}x)
\end{align*}
(and remember that $\rho_{s}=0$ for $s<0$ and $s>m$). In the above notation,
we obtain, after some calculations in Maple%
\begin{align*}
V_{s}^{(n,2)}  & =\frac{1}{4}\varepsilon\left(  \Gamma_{s}x,\Delta
^{2}x\right)  +\frac{1}{16}\left(  \Gamma_{s}x,x\right)  (x,\Delta x)+\frac
{1}{64}\varepsilon(\Gamma_{s}x,x)(x,x)^{2}+\frac{1}{16}(x,x)(\Gamma
_{s}x,\Delta x)\\
V_{s}^{(n,1)}  & =\frac{1}{4}\varepsilon(\Gamma_{s}x,\Delta x)+\frac{1}%
{16}\left(  \Gamma_{s}x,x\right)  (x,x)\\
V_{s}^{(n,0)}  & =\frac{1}{4}\varepsilon(\Gamma_{s}x,x)\\
V_{s}^{(n,-1)}  & =\frac{1}{4}\varepsilon\frac{(\Gamma_{s}x,\Delta^{-1}x)}%
{W}\\
V_{s}^{(n,-2)}  & =\frac{1}{W^{2}}\left(  \frac{1}{4}\varepsilon\left(
\Gamma_{s}x,\Delta^{-2}x\right)  +\frac{1}{16}\left(  \Gamma_{s-1}%
x,\Delta^{-1}x\right)  (x,\Delta^{-1}x)-\frac{1}{16}(\Gamma_{s-1}%
x,x)(x,\Delta^{-2}x)^{2}\right)
\end{align*}
For higher positive or negative $k$ these potentials quickly become very
complicated. Since $\Gamma_{1}=I$ and $\Gamma_{0}=0$ (due to (\ref{defrho}))
we see that%

\begin{align*}
V_{1}^{(n,2)}  & =\frac{1}{4}\varepsilon\left(  x,\Delta^{2}x\right)
+\frac{1}{8}(x,\Delta x)\left(  x,x\right)  +\frac{1}{64}\varepsilon
(x,x)^{3}\\
V_{1}^{(n,1)}  & =\frac{1}{4}\varepsilon(x,\Delta x)+\frac{1}{16}(x,x)^{2}\\
V_{1}^{(n,0)}  & =\frac{1}{4}\varepsilon(x,x)\\
V_{1}^{(n,-1)}  & =\frac{1}{4}\varepsilon\frac{(x,\Delta^{-1}x)}{W}\\
V_{1}^{(n,-2)}  & =\frac{1}{4}\varepsilon\frac{\left(  x,\Delta^{-2}x\right)
}{W^{2}}%
\end{align*}

This family of potentials has been obtained for the first time in
\cite{StefanMon} (see also \cite{stefan}). The potential $V_{1}^{(n,1)}$ is
the well known Garnier potential while $V_{1}^{(n,0)}$ is just harmonic
oscillator. Note that both in the Killing tensors $A_{s}$ and in the
potentials $V_{s}^{(n,k)}$ the sign $\varepsilon$ is present only at terms
with odd powers of $(x,x)$ which is clearly due to (\ref{xj}).

Let us now turn to the parabolic case $m=n-1$. In this case the structure of
the Killing tensors $A_{r}$ is more complicated

\begin{proposition}
For $m=n-1$ the tensors $A_{s}$ ($s=1,\ldots,n$) attain in the flat
coordinates $\left(  x_{1},\ldots,x_{n-1},r\right)  $ the form%
\begin{align*}
A_{s}^{ij}  & =\frac{1}{4}\frac{\partial^{2}\rho_{s-1}}{\partial\beta
_{i}\partial\beta_{j}}x_{i}x_{j}\text{, \ \ }i\neq j,\text{ }i,j=1,\ldots
,n-1\\
A_{s}^{ii}  & =-\varepsilon\frac{\partial\rho_{s}}{\partial\beta_{i}}-\frac
{1}{4}\sum\limits_{\substack{k=1 \\k\neq i}}^{n-1}\frac{\partial^{2}\rho
_{s-1}}{\partial\beta_{i}\partial\beta_{k}}x_{k}^{2}-\varepsilon\frac
{\partial\rho_{s-1}}{\partial\beta_{i}}r\text{, \ }i=1,\ldots,n-1\\
A_{s}^{in}  & =A_{s}^{ni}=\frac{1}{2}\frac{\partial\rho_{s-1}}{\partial
\beta_{i}}x_{i}\text{, \ }i=1,\ldots,n-1\\
A_{s}^{nn}  & =\rho_{s-1}%
\end{align*}
(again, with no summation over repeated indices) where $\rho_{s}=\rho
_{s}^{(n-1)}(\beta_{1},\ldots,\beta_{n-1})$ is given by (\ref{defrho}).
\end{proposition}

We will now investigate the potentials $V_{s}^{(n-1,k)}$. Let us slightly
change the notation:%
\begin{align*}
\Gamma_{s}  & =-\operatorname*{diag}\left(  \frac{\partial\rho_{s}^{(n-1)}%
}{\partial\beta_{1}},\cdots,\frac{\partial\rho_{s}^{(n-1)}}{\partial
\beta_{n-1}}\right)  \text{ \ \ \ }s=1,\ldots,n\\
\Delta & =\operatorname*{diag}(\beta_{1},\ldots,\beta_{n-1})\\
W  & =r+\frac{1}{4}\varepsilon\left(  x,\Delta^{-1}x\right)
\end{align*}
while $\left(  \cdot,\cdot\right)  $ stands now for the standard scalar
product in $\mathbf{R}^{n-1}$. We receive, after some calculations, again with
the help of Maple%
\begin{align*}
V_{s}^{(n-1,3)}  & =-\rho_{s-1}r^{3}+\frac{1}{4}\varepsilon\left(
\Gamma_{s-1}x,x\right)  r^{2}-\varepsilon\left(  \frac{1}{2}\left(  \Gamma
_{s}x,x\right)  -\frac{1}{4}\sum_{j=1}^{n-1}\rho_{s-j-1}(x,\Delta
^{j}x)\right)  r\\
& +\frac{1}{4}\varepsilon\left(  \Gamma_{s}x,\Delta x\right)  +\frac{1}%
{16}\left(  \Gamma_{s-1}x,x\right)  (x,x)\\
V_{s}^{(n-1,2)}  & =\rho_{s-1}r^{2}-\frac{1}{4}\varepsilon\left(  \Gamma
_{s-1}x,x\right)  r+\frac{1}{4}\varepsilon\left(  \Gamma_{s}x,x\right) \\
V_{s}^{(n-1,1)}  & =-\rho_{s-1}r+\frac{1}{4}\varepsilon\left(  \Gamma
_{s-1}x,x\right) \\
V_{s}^{(n-1,0)}  & =\rho_{s-1}\\
V_{s}^{(n-1,-1)}  & =\frac{1}{W}\left(  -\rho_{s-1}+\frac{1}{4}\varepsilon
\left(  \Gamma_{s-1}x,\Delta^{-1}x\right)  \right) \\
V_{s}^{(n-1,-2)}  & =\frac{1}{W^{2}}\left(  \rho_{s-1}+\frac{1}{4}%
\varepsilon\left(  \Gamma_{s}x,\Delta^{-2}x\right)  -\frac{1}{2}%
\varepsilon\left(  \Gamma_{s-1}x,\Delta^{-1}x\right)  +\frac{1}{4}%
\varepsilon\left(  \Gamma_{s-1}x,\Delta^{-2}x\right)  r\right. \\
& \left.  +\frac{1}{16}\left(  \Gamma_{s-2}x,\Delta^{-1}x\right)  \left(
x,\Delta^{-1}x\right)  -\frac{1}{16}\left(  \Gamma_{s-2}x,x\right)  \left(
x,\Delta^{-2}x\right)  \right)
\end{align*}
(with $s=1,\ldots,n$) and again these formulas become quickly very complicated
for higher positive or negative $k$. Let us now specify these potentials for
case $s=1$. Since $\rho_{0}=1$, $\Gamma_{1}=I$ while $\Gamma_{0}=0$ we get%
\begin{align*}
V_{1}^{(n-1,4)}  & =r^{4}+\frac{3}{4}\varepsilon(x,x)r^{2}-\frac{1}%
{2}\varepsilon(x,\Delta x)r+\frac{1}{4}\varepsilon(x,\Delta^{2}x)+\frac{1}%
{16}(x,x)^{2}\\
V_{1}^{(n-1,3)}  & =-r^{3}+\frac{1}{2}\varepsilon\left(  x,x\right)
r+\frac{1}{4}\varepsilon\left(  x,\Delta x\right) \\
V_{1}^{(n-1,2)}  & =r^{2}+\frac{1}{4}\varepsilon\left(  x,x\right) \\
V_{1}^{(n-1,1)}  & =-r\\
V_{s}^{(n-1,0)}  & =1\\
V_{s}^{(n-1,-1)}  & =-\frac{1}{W}\\
V_{s}^{(n-1,-2)}  & =\frac{1+\frac{1}{4}\varepsilon\left(  x,\Delta
^{-2}x\right)  }{W^{2}}%
\end{align*}
Again, in the above formulas the sign $\varepsilon$ is present only at terms
with odd powers of $(x,x)$.

For arbitrary $0\leq m\leq n-2$ the form of the Killing tensors $A_{s}$ is not
so transparent and we will omit it here. Let us however present some results
on separable potentials $V^{(m,k)}$ in case $0\leq m<n-1$. In the case $m=0$
the variables are $(r_{1},\ldots,r_{n})$ and as we mentioned above,
$V_{r}^{(0,k)}=U_{r}^{(k)}$ where $U_{r}^{(k)}$ are polynomial (for $k\geq n$)
or rational (for $k\leq0$) functions of $q_{i}$ given by (\ref{U})-(\ref{R}).
Thus (we remember that for $m=0$ we have $q_{i}=a_{i}$, see (\ref{GF0}))%
\[
V^{(0,k)}(r)=U^{(k)}(r)=\left(
\begin{array}
[c]{cccc}%
-a_{1} & 1 &  & \\
-a_{2} &  & \ddots & \\
\vdots &  &  & 1\\
-a_{n} & 0 & \cdots & 0
\end{array}
\right)  ^{k}\left(
\begin{array}
[c]{c}%
0\\
\vdots\\
0\\
1
\end{array}
\right)  \text{, \ }k\in\mathbf{Z}%
\]
with%
\[
a_{i}=r_{i}+\frac{1}{4}\sum\limits_{j=1}^{i-1}r_{j}r_{i-j},\text{
\ \ }i=1,\ldots,n
\]
and so the first nontrivial potential is $V^{(0,n)}=(-a_{1},\ldots,-a_{n}%
)^{T}$. The situation is much more complex for the arbitrary $m$ such that
$0<m<n-1$. Before we present some results in this generic case, let us
introduce a notation similar to used above for the cases $m=n$ and $m=n-1$. We
denote%
\begin{align*}
\Gamma_{s}  & =-\operatorname*{diag}\left(  \frac{\partial\rho_{s}^{(m)}%
}{\partial\beta_{1}},\cdots,\frac{\partial\rho_{s}^{(m)}}{\partial\beta_{m}%
}\right)  \text{ \ \ \ }s=1,\ldots,m\\
\Delta & =\operatorname*{diag}(\beta_{1},\ldots,\beta_{m})
\end{align*}
and to shorten the notation we will simply denote $\rho_{s}^{(m)}=\rho
_{s}^{(m)}(\beta_{1},\ldots,\beta_{m})$ by $\rho_{s}$. The variables are now
$(x_{1},\ldots,x_{m},r_{1}\ldots,r_{n-m})$ or simply $(x,r)$. This time
$\left(  \cdot,\cdot\right)  $ will denote the scalar product in
$\mathbf{R}^{n-m}$. Introduce now the column vector of potentials
$U^{(k)}=U^{(k)}(r_{1},\ldots.r_{n-m})$, $s=1,\ldots,n-m$ given by
\[
U^{(k)}(r)=\left(
\begin{array}
[c]{cccc}%
-a_{1}(r) & 1 &  & \\
-a_{2}(r) &  & \ddots & \\
\vdots &  &  & 1\\
-a_{n-m}(r) & 0 & \cdots & 0
\end{array}
\right)  ^{k}\left(
\begin{array}
[c]{c}%
0\\
\vdots\\
0\\
1
\end{array}
\right)  \text{, \ }k\in\mathbf{Z}%
\]
with
\[
a_{i}=r_{i}+\frac{1}{4}\sum\limits_{j=1}^{i-1}r_{j}r_{i-j},\text{
\ \ }i=1,\ldots,n-m
\]
(so that the last trivial potential is $U^{(n-m-1)}=(1,\ldots,0)^{T}$). We
obtain, after some Maple calculations%
\begin{align*}
& \vdots\\
V_{s}^{(m,n-m+2)}  & =\sum_{j=1}^{n-m}\rho_{s-j}U_{j}^{(n-m+2)}+\frac{1}%
{4}\varepsilon\sum_{j=0}^{2}\left(  \Gamma_{s-(n-m-j)}x,x\right)
U_{1}^{(n-m+1-j)}\\
V_{s}^{(m,n-m+1)}  & =\sum_{j=1}^{n-m}\rho_{s-j}U_{j}^{(n-m+1)}+\frac{1}%
{4}\varepsilon\sum_{j=0}^{1}\left(  \Gamma_{s-(n-m-j)}x,x\right)
U_{1}^{(n-m-j)}\\
V_{s}^{(m,n-m)}  & =\sum_{j=1}^{n-m}\rho_{s-j}U_{j}^{(n-m)}+\frac{1}%
{4}\varepsilon\left(  \Gamma_{s-(n-m)}x,x\right)  U_{1}^{(n-m-1)}\\
V_{s}^{(m,n-m-1)}  & =\rho_{s-1}U_{1}^{(n-m-1)}=\text{const.}\\
& \vdots\\
V_{s}^{(m,-1)}  & =\frac{-\sum_{j=1}^{n-m}\rho_{s-j}a_{j-1}+\frac{1}%
{4}\varepsilon(\Gamma_{s-(n-m)}x,\Delta^{-1}x)}{a_{n-m}+\frac{1}{4}%
\varepsilon(x,\Delta^{-1}x)}\\
& \vdots
\end{align*}

Majority of these potentials seem to be new. Potentials higher than
$V_{s}^{(m,n-m+2)}$ as well as lower than $V_{s}^{(m,-1)}$contain terms at
least quadratic in $(x,x)$ and are too complicated to present it here.

\section{Flat coordinates for St\"{a}ckel systems in the complex case}

We will now investigate the case of complex conjugate roots in the polynomial
$B_{m}(\lambda)$. Assume thus that the first $2p$ ($2p\leq m$) roots
$\beta_{j}$ in $B_{m}(\lambda)$ are pairwise complex conjugate with nonzero
imaginary parts:%
\[
\overline{\beta_{2r-1}}=\beta_{2r}\text{, }r=1,\ldots,p,\text{ }%
\operatorname{Im}(\beta_{r})\neq0
\]
It is easy to check that $x_{j}^{2}$ given by (\ref{xj}) are then pairwise
complex conjugate as well: $\overline{x_{2r-1}}=\pm x_{2r}$, $r=1,\ldots,p$
although the generating function (\ref{GF})%
\[%
{\displaystyle\sum\limits_{j=0}^{n-m}}
z^{n-m-j}a_{j}-\frac{1}{4}\varepsilon\sum_{j=1}^{m}\frac{x_{j}^{2}}%
{z-\beta_{j}}\equiv\frac{%
{\displaystyle\prod\limits_{j=1}^{n}}
(z-\lambda_{j})}{%
{\displaystyle\prod\limits_{j=1}^{m}}
(z-\beta_{j})}%
\]
remains real as $B_{m}(\lambda)=%
{\displaystyle\prod\limits_{j=1}^{m}}
(z-\beta_{j})$ is real and since%
\[
\frac{x_{2r-1}^{2}}{z-\beta_{2r-1}}+\frac{x_{2r}^{2}}{z-\beta_{2r}}%
\]
is real for any $r=1,\ldots,p$. Let us now define new real variables%
\begin{equation}
\eta_{2s-1}=\frac{x_{2s-1}+x_{2s}}{\sqrt{2}},\eta_{2s}=\frac{x_{2s-1}-x_{2s}%
}{\sqrt{2}i},\text{ \ \ }s=1,\ldots,p\label{real}%
\end{equation}
The transformation inverse to (\ref{real}) is%

\[
\text{ }x_{2s-1}=\frac{1}{\sqrt{2}}\left(  \eta_{2s-1}+i\eta_{2s}\right)
\text{, }x_{2s}=\frac{1}{\sqrt{2}}\left(  \eta_{2s-1}-i\eta_{2s}\right)
,\text{ \ \ }s=1,\ldots,p
\]
We are now in position to formulate a theorem analogous to Theorem
\ref{mainthm}.

\begin{theorem}
\bigskip The metric $G$ defined by (\ref{GBen}) attains in coordinates
$(\eta,x,r)=\left(  \eta_{1},\ldots,\eta_{2p},x_{2p+1},\ldots,x_{m}%
,r_{1},\ldots r_{n-m}\right)  $ the form%
\[
G=\left(
\begin{array}
[c]{ccc}%
\varepsilon D_{2p\times2p} & 0 & 0\\
0 & \varepsilon I_{(m-2p)\times(m-2p)} & 0\\
0 & 0 & J_{(n-m)\times(n-m)}%
\end{array}
\right)
\]
where $D_{2p\times2p}$ is a $2p\times2p$ diagonal matrix with intertwined
entries $1$ and $-1$:
\[
D_{2p\times2p}=\text{diag}(1,-1,\ldots,1,-1)
\]
and where as before $I_{k\times k}$ denotes the $k\times k$ indetity matrix
and $J_{k\times k}$ denotes the $k\times k$ matrix given by $\left(
J_{k\times k}\right)  _{ij}=\delta_{i,k-j+1}$ i.e. with entries equal to zero
everywhere except on the antidiagonal where all the entries are equal to $1$.
As before, the symbol $0$ in the above formula denotes a zero matrix of
appropriate dimensions.
\end{theorem}

\begin{proof}
One can prove this theorem similarly as one proves Theorem \ref{mainthm}. As
before, we observe that%
\[
G=\sum_{k=0}^{m}\rho_{k}^{(m)}G_{k}%
\]
(with the real coefficients $\rho_{k}^{(m)}$). The form of tensors $G_{k}$ in
Vi\`{e}te coordinates is known \cite{blaser} so one can easily first transform
the tensor $G$ to the Vi\`{e}te coordinates, and then use the inverse of the
Jacobian of the map (\ref{xrV}) to transform it to the variables $(\eta,x,r)$.
\end{proof}

Note that the map (\ref{xrV})%
\begin{equation}
q_{i}=\sum_{j=0}^{n-m}\rho_{i-j}a_{j}+\frac{1}{4}\varepsilon\sum_{j=1}%
^{m}\frac{\partial\rho_{i-(n-m)}}{\partial\beta_{j}}x_{j}^{2}\text{,
\ }i=1,\ldots,n\label{mapa}%
\end{equation}
is actually a real map even in the complex case. The formal (algebraic, not
complex-analytic) derivatives of $\rho_{i-(n-m)}$ with respect to those of
$\beta_{j}$ that are complex, together with the corresponding complex
$x_{j}^{2}\,$, enter the second sum in (\ref{mapa}) in complex conjugate
pairs, so that this sum is indeed real.

Thus, the variables $(\eta,x,r)=\left(  \eta_{1},\ldots,\eta_{2p}%
,x_{2p+1},\ldots,x_{m},r_{1},\ldots r_{n-m}\right)  $ are flat but not
orthogonal while the original variables $(\lambda_{1},\ldots,\lambda_{n})$ are
orthogonal but not flat.

\begin{corollary}
The signature $(n_{+},n_{-})$ of the metric $G$ is in the complex case given
by%
\begin{align*}
\left(  n_{+},n_{-}\right)   & =\left(  n-\left[  \frac{n-m}{2}\right]
-p,\left[  \frac{n-m}{2}\right]  +p\right)  \text{ \ in the region where
}\varepsilon=+1\\
\left(  n_{+},n_{-}\right)   & =\left(  n-m-\left[  \frac{n-m}{2}\right]
+p,\left[  \frac{n-m}{2}\right]  +m-p\right)  \text{ \ in the region where
}\varepsilon=-1
\end{align*}
where $\left[  \alpha\right]  $ denotes the integer part of the number
$\alpha$.
\end{corollary}

It is not possible to write down Killing tensors in the complex case in any
reasonably compact form even for the elliptic case, the subcases become too
many, albeit it is not difficult to calculate these tensors for any given
choice of parameters $n,m,p$. There is on the other hand no change to the
potentials $V_{s}^{(m,k)}$ in the complex case as the function $\sigma
(\lambda)=-\lambda^{k}B_{m}(\lambda)$ in (\ref{BenSC}) is real even in the
complex case. Of course the explicit form of the potentials do change. For
example, in the complex elliptic case, in the variables $(\eta,x,r)=\left(
\eta_{1},\ldots,\eta_{2p},x_{2p+1},\ldots,x_{n}\right)  $, the potentials
$V_{s}^{(n,k)}$, given in the previous section, do change its explicit form.
The scalar product $(x,x)$ (which is real due to the fact that now
$\overline{x_{2r-1}}=\pm x_{2r}$, $r=1,\ldots,p$) is given explicitly as%
\[
(x,x)=\sum\limits_{s=1}^{p}\left(  \eta_{2s-1}^{2}-\eta_{2s}^{2}\right)
+\sum\limits_{s=2p+1}^{n}\eta_{s}^{2}%
\]
and so on. The fact that $V_{s}^{(n,k)}$ are real also in the complex case can
be easily seen, as all the expressions $(x,x)$, $(x,\Delta^{p}x)$ for
$p\in\mathbf{Z}$, $(\Gamma_{s}x,x)$ and $(\Gamma_{s}x,\Delta^{p}x)$ are real.
For example, $(\Gamma_{s}x,x)$ contains apart from the real terms also pairs
of complex terms of the form%
\[
-\frac{\partial\rho_{s}^{(n)}}{\partial\beta_{2i-1}}x_{2i-1}^{2}%
-\frac{\partial\rho_{s}^{(n)}}{\partial\beta_{2i}}x_{2i}^{2}%
\]
(where again we have formal i.e. non-complex algebraic derivatives). We easily
see that each such pair and so the whole expression $(\Gamma_{s}x,x)$ is real.

\section{Conclusions}

In this article we presented flat coordinates for majority of flat St\"{a}ckel
systems. These coordinates coincide in the elliptic (i.e. when $m=n$) and in
the parabolic ($m=n-1$) cases with the well known generalized Jacobi elliptic
(respectively parabolic) coordinates. The only St\"{a}ckel systems that have
not been covered by our construction are those that are generated by the
separation curve (\ref{BenSC}) in the case when some of the roots of the
polynomial $B_{m}(\lambda)$ has algebraic multiplicity larger than $1$. On the
other hand, in \cite{blaser} the authors considered the completely degenerated
case i.e. when all $\beta_{i}$ coincide (and are equal to zero). For the case
$m=0$ our formulas coincide with the formulas obtained there.

In this paper we also presented - in the elliptic $m=n$ and parabolic $m=n-1$
case - a compact form of Killing tensors of St\"{a}ckel systems in these flat
coordinates and also a compact form of their separable potentials
$V_{s}^{(m,k)}$. They contain the well known Garnier system and they encompass
for example both families of separable potentials found in \cite{StefanMon}.
We also presented the form of the first few potentials $V_{s}^{(m,k)}$ in the
case $m=0$ i.e. the formulas for $V^{(0,k)}(r)=U^{(k)}(r)$ in flat coordinates.

\section{Acknowledgement}

Both authors were partially supported by The Royal Swedish Academy of Sciences
grant no FOA13Magn-088.

\end{document}